\documentclass[11pt]{article}
\usepackage[a4paper,margin=0.8in]{geometry}
\usepackage[T1]{fontenc}
\usepackage{lmodern}
\usepackage[utf8]{inputenc}
\usepackage{graphicx, xcolor, colortbl}
\usepackage{tabularx}
\usepackage{multirow}
\usepackage{tikz}
\usetikzlibrary{arrows,shapes}
\usetikzlibrary{arrows.meta}
\usepackage{dsfont}
\usepackage{amsmath, amsthm, amssymb, mathtools, nicefrac}
\usepackage{enumitem}
\usepackage{float}
\usepackage{upgreek}
\usepackage{hyperref}
\usepackage{parskip}
\usepackage{subcaption}
\usepackage{stmaryrd}
\usepackage{tcolorbox}
\usepackage{booktabs}
\usepackage{multicol}

\definecolor{grammarGreen}{RGB}{0,120,100}

\newcommand{\gG}[1]{\textcolor{grammarGreen}{\texttt{#1}}}

\providecommand{\keywords}[1]
{
  \small	
  \textbf{\textit{Keywords}} #1
}

\newtheorem{thm}{Theorem}

\theoremstyle{definition}
\newtheorem{definition}{Definition}

\newtheorem{example}{Example}

\font\srs=cmr6
\font\ssi=cmti8
\font\sst=cmtt8

\title{A-COMPASS: Formal Foundations for Anonymity Analysis in Microdata}
\date{}

\author{\frenchspacing
\bf Tamara Tagliavia$^{1}$, Silvia Ghilezan$^{1,2}$\\\\
$^{\srs 1}${\ssi Mathematical Institute of the Serbian Academy of Sciences and Arts, Belgrade, Serbia} \\
$^{\srs 2}${\ssi Faculty of Technical Sciences, University of Novi Sad, Novi Sad, Serbia} \\
\\
{\ssi {\ssi E-mail}:}
$^{\srs 1}${\sst tamara.stefanovic@mi.sanu.ac.rs},\ $^{\srs 2}${\sst gsilvia@uns.ac.rs}
}

\begin{document}

\maketitle
\begin{abstract}
In the information age, one of the leading problems is how to ensure individual's privacy. Depending on the context in which privacy is considered, various data privacy models have emerged. However, the domain of formal verification of these models is still not sufficiently explored even when it comes to the most basic models. An attempt to verify privacy requirements is the Compliance Assertion Language (COMPASS). In COMPASS, one can specify an anonymity condition that a table needs to satisfy, and an action that will modify the table if the condition is not satisfied. It is designed to operate on preprocessed tables in a form one record - one group of people. In this paper, we modify the COMPASS language in order to operate on microdata tables in their usual form of one record - one person. The modified language is called A-COMPASS. Along with checking of previously applied anonymity conditions, A-COMPASS enables the execution of anonymization actions as a new feature. We further provide the syntax and the semantics for the A-COMPASS language. We also prove the most important properties of the introduced semantics like determinism and compositionality. Finally, we provide a mechanism to verify anonymity properties, such as k-anonymity and l-diversity. 
\end{abstract}


\section{Introduction}
\label{intro}
On a daily basis, a huge volume of data is collected from a variety of devices and platforms. Such data captures human behaviors, routines, activities, etc. Although the ability to manipulate large amount of data has improved our lives in many domains, it has also evolved new problems. The most important among them is certainly the \textbf{data privacy} problem.

The first definition of privacy was given by Warren and Brandeis in 1890. in their “The Right to Privacy” \cite{warren1890right}, long before the information age. They define privacy as the right to be left alone. Nowadays, privacy is considered a fundamental right. Countries around the world have realized the need to protect the privacy of their citizens. The most important data protection legislation enacted to date is the \textbf{General Data Protection Regulation} (GDPR) \cite{GDPR}. The GDPR is a law originating from the EU, yet it applies to businesses all over the world that collect and process personal information of EU citizens. Such businesses are required to publish and abide a privacy policy agreement. A privacy policy is a statement that details how personal data provided by users will be collected, stored, processed, and shared with third parties. They are usually written in natural language, and thus their formal verification is very demanding.

In parallel with the development of technology and the legal aspect of privacy, technologies for privacy protection also began to develop. In the core of these technologies are mathematical models and formal methods. One of the first models presented for data privacy is \textbf{k-anonymity} \cite{Samarati98, Sweene02}, which deals with a special branch of privacy, called anonymity. It provides certain conditions that a \textbf{microdata} table has to fulfill in order to mitigate the risk of identity disclosure. A microdata record is considered to be a record that is associated with information of one individual. K-anonymity requires that one person in a microdata set cannot be distinguished from at least k-1 other persons with similar characteristics. The formal definition of k-anonymity is revisited in Section~\ref{anonymization}. An improvement of this model is called \textbf{l-diversity} \cite{MKGV07}, and it deals with attribute disclosure, along with identity disclosure. 

When it comes to data privacy models, \textbf{differential privacy} \cite{Dwork2006}, must be mentioned. Differential privacy deals mainly with privacy in aggregated statistics by adding noise to the query result. Therefore, it is a property of a mechanism (that adds noise), not a property of a table like the previous models. 

The research on mathematical models for privacy is increasingly developing, and there is an urge for \textbf{formal verification} of these models. Existing research on verification, such as \cite{hoare2015, dfuzz2013, fuzz2010}, mainly deals with differential privacy, including sensitivity analysis over relational algebra \cite{CatMar2012}. In turn, basic models from the domain of anonymity still seek for more detailed formalization - a gap that formal methods in security suggest to be closed \cite{BUGLIESI2017}. The basic anonymity models are widely used, and in order to be sure that certain privacy laws are adequately applied in practice, we need formal verification of these models.

One of the pioneers in the verification of k-anonymity and l-diversity is a formal language called the \textbf{Compliance Assertion Language} (COMPASS) introduced in \cite{GOKI2023}. COMPASS is an SQL-based language designed to formulate anonymity requirements. Each COMPASS requirement is defined as an assertion followed by an action. An assertion represents the anonymity condition that the table has to satisfy, whereas the action represents the modification that has to be performed on the table if the assertion is not satisfied. The paper focuses on a special kind of tables that is preprocessed and in which each record represents a group of people. In a usual microdata table, each record represents one person. This limitation can be easily overcome by a slight extension of the COMPASS syntax. In this paper, we modify the syntax of COMPASS and we introduce it's semantics.

\paragraph{Contributions} The main contributions of the paper are:

\begin{enumerate}

    \item We extend the COMPASS language into a new language called A-COMPASS (Anonymity Compliance Assertion Language) which can operate with tables in the form of one record - one person in addition to tables in the form one record - one group of people. This modification also allows the performance of anonymization actions(Section~\ref{A-COMPASS}).
    
    \item We modify the syntax of the COMPASS language by adding a new aggregation operation COUNT DISTINCT, removing the action JOIN, and extending the action ZERO into REPLACE(Section~\ref{A-COMPASS}). 
    
    \item We provide a semantics for the A-COMPASS language(Section~\ref{semantics}).
    
    \item We prove fundamental properties of the introduced semantics, such as determinism and compositionality. We also validate anonymity properties of A-COMPASS using the introduced semantics(Section~\ref{properties}).
\end{enumerate}

\section{Anonymization in Microdata}\label{anonymization}
A record that contains information related to a specific individual (a citizen
or a company) is called \textbf{microdata}. The representation of microdata is usually in the form of a table, where each row represents a different individual, while each column contains information about the collected attributes. Therefore, a microdata table can be considered a special case of an SQL table. Table~\ref{table: Non-anonymized microdata table} is an example of a microdata table related to the annual electricity consumption. 

This kind of data is widely collected from surveys, administrative systems, or transactional processes and then forwarded to experts for analysis or published in the form of public statistics. There is a wide spectrum of applications of microdata, from medicine, economics, and social sciences to machine learning. However, microdata releases are very challenging for individual privacy \cite{anonymitybook}.

\begin{table}[h!]
\centering
    \begin{minipage}{0.45\linewidth}
    \centering
\caption{Non-anonymized microdata table}
\label{table: Non-anonymized microdata table}
\begin{tabular}{cccc}
\toprule
Record ID & Age & Postal Code & AEC (kWh) \\
\midrule
1  & 54  & 21201 & 2200 \\
2  & 54  & 21201 & 7400 \\
3  & 54  & 21203 & 8600 \\
4  & 82  & 21410 & 10500 \\
5  & 86  & 21410 & 3500 \\
6  & 83  & 21410 & 8600 \\
7  & 36  & 21101 & 4800 \\
8  & 36 & 21102 & 4800 \\
9  & 45 & 21101 & 6200 \\
10 & 45  & 21102 & 5400 \\
\bottomrule
\end{tabular}
\end{minipage}
\hfill
\begin{minipage}{0.45\linewidth}
    \centering
    \caption{Suppression method}
\label{table: suppression method}
    \begin{tabular}{cccc}
\toprule
Record ID & Age & Postal Code & AEC (kWh) \\
\midrule
1  & 54  & 21201 & 2200 \\
2  & 54  & 21201 & 7400 \\
3  & 54  & 21203 & 8600 \\
4  & \textbf{80}  & 21410 & 10500 \\
5  & \textbf{80} & 21410 & 3500 \\
6  & \textbf{80}  & 21410 & 8600 \\
7  & 36  & 21101 & 4800 \\
8  & 36 & 21102 & 4800 \\
9  & 45 & 21101 & 6200 \\
10 & 45  & 21102 & 5400 \\
\bottomrule
\end{tabular}

\end{minipage}
    
\end{table}

\textbf{Attribute Types.} Attributes in microdata are classified into four categories: identifiers, quasi-identifiers, sensitive attributes, and non-sensitive attributes. An identifier uniquely identifies a person (e.g., social security number, a passport number). A quasi-identifier cannot uniquely identify a person; however, the combination of quasi-identifiers can lead to identification (e.g., name, age, profession, etc.). Sensitive attributes carry  sensitive information about the individuals in the dataset (e.g., sex orientation, health condition, etc.). Non-sensitive attributes are those that don't belong to any of the previous categories.\\

Table~\ref{table: Non-anonymized microdata table} does not contain an identifier, since identifiers are the first to be removed when it comes to microdata publishing. Attributes Age and Postal Code are quasi-identifiers, whereas Annual Electricity Consumption (AEC) is the sensitive attribute.\\

\textbf{Disclosure Risk.} When it comes to anonymity and privacy issues in microdata releases, there are two associated disclosure risks: risk of an identity disclosure and risk of an attribute disclosure \cite{Hundepool2012}. Identity disclosure occurs when an adversary links a record in the data with a specific individual. Then the values of the published attributes are also linked with that individual. Attribute disclosure occurs when an adversary can infer a value of a sensitive attribute for a specific individual with high probability. It can occur even without an identity disclosure. 

For example, if adversaries know that in the place with postal code 21410 lives only one person who is 86 years old (postal code 21410 represents a rural area), they can infer that the record with Record ID = 5 corresponds to that person. This is an example of an identity disclosure. Further, by looking at the value for AEC for this record, the adversary can conclude that this person lives alone (since 3500 kWh consumption per year does not represent high consumption). 

Moreover, if adversaries want to infer the AEC of a person who is 36 years old and lives in the town (the exact postal code is unknown, yet it is known that it is either 21101 or 21102), they cannot conclude which record belongs to that person (two records with Record ID = 7 and Record ID = 8 match). Thus, identity disclosure is not possible. However, since the value of AEC is the same for both records, they can learn its value with probability 1. \\

\textbf{Anonymization Methods.} In order to mitigate disclosure risks, two base anonymization methods were introduced: suppression and generalization. Suppression is a method where certain values of the attributes are usually replaced by an asterisk, e.g., '$\star$', \cite{BayardoA05, Samarati98}. Replacing the attribute value with some other arbitrary value can also be considered a suppression method. Generalization is a method where all values of the attribute are replaced by a more general value \cite{Aggarwal04, LeFevreDR05}. In this way, it can be considered that the values of that attribute are divided into categories.

\begin{example}[Suppression] Suppression is often used to cover rare values, since they represent a high risk for re-identification, e.g., ages over 80. When it comes to Table~\ref{table: Non-anonymized microdata table}, if researchers conclude that for the data analysis an age category (not exact age) is enough, and if they want to protect the privacy of older individuals, they can decide to replace all values of an attribute Age that are over 80 with the value 80. Table~\ref{table: suppression method} is an anonymized version of Table~\ref{table: Non-anonymized microdata table} obtained by the suppression method.

\label{example: suppression}
\end{example}

\begin{example}[Generalization] For  statistical analysis, often knowing the area is sufficient, meaning that it is not necessary to know the exact town. That is why the values of the attribute Postal Code can be generalized by removing the last two digits. The new values of the attribute Postal Code now represent specific areas. Table~\ref{table: generalization method} is an anonymized version of Table~\ref{table: Non-anonymized microdata table} obtained by the generalization method. 

\begin{table}[h!]
\centering
    \begin{minipage}{0.45\linewidth}
    \centering
\caption{Generalization method}
\label{table: generalization method}
\begin{tabular}{cccc}
\toprule
Record ID & Age & Postal Code & AEC (kWh) \\
\midrule
1  & 54  & \textbf{212**} & 2200 \\
2  & 54  & \textbf{212**} & 7400 \\
3  & 54  & \textbf{212**} & 8600 \\
4  & 82  & \textbf{214**} & 10500 \\
5  & 86  & \textbf{214**} & 3500 \\
6  & 83  & \textbf{214**} & 8600 \\
7  & 36  & \textbf{211**} & 4800 \\
8  & 36 & \textbf{211**} & 4800 \\
9  & 45 & \textbf{211**} & 6200 \\
10 & 45  & \textbf{211**} & 5400 \\
\bottomrule
\end{tabular}

\end{minipage}
\hfill
\begin{minipage}{0.45\linewidth}
    \centering
    \caption{2-anonymity and 1-diversity}
\label{table: k-anonymity and l-diversity}
   \begin{tabular}{cccc}
\toprule
Record ID & Age & Postal Code & AEC (kWh) \\
\midrule
1  & 54  & 212** & 2200 \\
2  & 54  & 212** & 7400 \\
3  & 54  & 212** & 8600 \\
\hline
4  & 80  & 214** & 10500 \\
5  & 80  & 214** & 3500 \\
6  & 80  & 214** & 8600 \\
\hline
7  & 36  & 211** & 4800 \\
8  & 36 & 211** & 4800 \\
\hline
9  & 45 & 211** & 6200 \\
10 & 45  & 211** & 5400 \\
\bottomrule
\end{tabular} 

\end{minipage}
    
\end{table}

\label{example: generalization}
\end{example}

The first data privacy models are based on these two anonymization methods. One of them is called \textbf{k-anonymity} and it was introduced by Samarati and Sweeney in \cite{Samarati98, Sweene02}. The second one, called \textbf{l-diversity}, is essentially the extension of k-anonymity. It was introduced by Machanavajjhala et al. in \cite{MKGV07} with the aim of improving the weaknesses of the previous model. Below we will recall the definitions of both models.

\begin{definition}[k-anonymity] Let $R$ be a table with a set of attributes $A$. Let $Q_R\subseteq A$ be a set of quasi-identifiers for  table $R$. The table $R$ satisfies the k-anonymity principle if and only if  each combination of values of the attributes in $Q_R$ is shared by k or more records. 
    
\end{definition}

The k-anonymity principle says that an individual cannot be distinguished from at least k-1 other individuals that share the same combination of values of quasi-identifiers. The records of these individuals form an equivalence class. In this way k-anonymity provides protection against identity disclosure, up to some level k.

However, k-anonymity is vulnerable to the \textit{homogeneity attack} and the \textit{background knowledge attack}. Both attacks refer to attribute disclosure. The homogeneity attack can be performed when all the values of the sensitive attribute in an equivalence class are the same. If an adversary concludes that someone's data is in that equivalence class, they can exactly predict the sensitive value for that person. The background knowledge attack can be performed when the attacker uses some side (background) knowledge to reduce the set of possible values for the sensitive attribute. In order to overcome these attacks, the l-diversity model was developed \cite{MKGV07}.

\begin{definition}[l-diversity] An equivalence class of a table $R$ satisfies the l-diversity principle if and only if it contains at least l well-represented values for the sensitive attribute. The table $R$ satisfies the l-diversity principle if and only if each equivalence class of $R$ satisfies the l-diversity principle.
    
\end{definition}

\begin{example}[k-anonymity and l-diversity] By applying both suppression and generalization methods from Example~\ref{example: suppression} and Example~\ref{example: generalization}, Table~\ref{table: k-anonymity and l-diversity} is obtained. This table has 4 equivalence classes. It satisfies the k-anonymity principle for k=2, since the smallest equivalence classes have 2 records. It satisfies the l-diversity principle for l=1, since the values for the sensitive attribute in the third equivalence class are all the same. This class is also vulnerable to the homogeneity attack.

\end{example}

\section{A-COMPASS Language}\label{A-COMPASS}
In order to verify whether a certain dataset satisfies anonymity requirements, the Compliance Assertion Language (COMPASS) was introduced in \cite{GOKI2023}. In COMPASS one can formulate the anonymity conditions that the dataset needs to satisfy and, if a condition is not satisfied, actions that need to be performed to modify the dataset to meet the defined anonymity criteria. We modify the COMPASS language in order to improve the anonymity analysis. Our goal is to extend the COMPASS language twofold:

\begin{itemize}
    \item to allow users to work with standard microdata tables in the form of one record - one person;
    \item to allow users to anonymize a table by suppression and generalization.
\end{itemize}

\subsection{A-COMPASS Syntax}

The syntax of the A-COMPASS language is shown in Figure~\ref{figure: A-COMPASS syntax}. We have introduced the commands in red. The rest of the commands are inherited from COMPASS language. \\


In the following examples, we explain the A-COMPASS commands and how they are used to modify Table~\ref{table: Non-anonymized microdata table} and to perform anonymity analysis.

\begin{example}[Suppression in A-COMPASS]\label{example: suppression} In order to obtain Table~\ref{table: suppression method} from Table~\ref{table: Non-anonymized microdata table}, the following A-COMPASS requirement needs to be applied.

\begin{tcolorbox}[colback=gray!10,colframe=gray!10, boxrule=0.6pt, arc=0pt]

\begin{tabular}{rl}

   \gG{EACH} & \gG{RESULT}\\
   : & Age $\leq$ 80\\
   : & \gG{REPLACE} Age \gG{WITH} \.80;
   \end{tabular}

\end{tcolorbox}

The assertion \gG{EACH RESULT}: Age $\leq$ 80 represents the condition we want our table to satisfy. The command EACH RESULT means that the assertion affects all records, Age $\leq$ 80 is the Boolean condition that has to be satisfied. If the assertion is not satisfied, action REPLACE Age WITH 80 replaces the values of the attribute Age with 80 in each record where Age > 80.

NOTE: If the command SOME RESULT is used instead of EACH RESULT, only a subset of records is required to satisfy the defined condition, not necessarily all records. If the assertion is not satisfied in this case, action REPLACE Age WITH 80 will affect all records. 
    
\end{example}

\begin{minipage}{\linewidth}

\hrule
\vspace{0.5em}

\begin{tabular}{lcl}

\gG{requirements} & ::= & \{ \gG{assertion} ':' \gG{action} ';' \} \\[6pt]

\gG{assertion} & ::= & ( 'EACH' \gG{result} ':' \gG{boolExpr} ) \\
&  & \textbar( 'SOME' \gG{result} ':' \gG{boolExpr} ) \\[6pt]

\gG{result} & ::= & 'RESULT' \textbar \gG{resultQuery} \\

\gG{resultQuery} & ::= & \gG{processQuery} \textbar \gG{filterQuery} \\

\gG{processQuery} & ::= & 'PROCESS' \gG{generalization} \\
& & [ 'WHERE' \gG{boolExpr} ] \\
& & [ 'GROUP\_BY' \gG{attribute} \{ ',' \gG{attribute} \} ] \\[6pt]

\gG{generalization} & ::= & ( ( 'SUM' \textbar 'MIN' \textbar 'MAX' ) '(' \gG{attribute} ')' ) \\
& & \textbar( 'COUNT' '(' '*' ')' ) \textbar {\color{red}( 'COUNT DISTINCT' '(' 'attribute') )} 'AS' \gG{attribute} \\[6pt]

\gG{filterQuery} & ::= & 'FILTER' \gG{boolExpr} \\[8pt]

\gG{boolExpr} & ::= & \gG{boolProduct} \{ 'OR' \gG{boolProduct} \} \\

\gG{boolProduct} & ::= & \gG{boolLiteral} \{ 'AND' \gG{boolLiteral} \} \\

\gG{boolLiteral} & ::= & [ 'NOT' ] \gG{boolCond} \\
& & \textbar'(' \gG{boolExpr} ')' \\

\gG{boolCond} & ::= & \gG{attribute} ( '$<$' \textbar '$>$' \textbar '$=$' \textbar '$\leq$' \textbar '$\geq$' )( number \textbar string ) \\[10pt]

\gG{action} & ::= & 'REJECT' \\
& & \textbar{\color{red}( 'REPLACE' attribute 'WITH' value )} \\
& & \textbar( 'RANDOM' \gG{attribute} \gG{lowerBound} \gG{upperBound} ) \\[8 pt]

\gG{attribute} & ::= & string \\
\gG{value} & ::= & string \\
\gG{lowerBound} & ::= & number \\
\gG{upperBound} & ::= & number \\

\end{tabular}

\vspace{0.5em}
\hrule
\vspace{0.5em}

\captionof{figure}{The Syntax of the A-COMPASS language.}
\label{figure: A-COMPASS syntax}

\end{minipage}

\begin{example}[Generalization in A-COMPASS]\label{example: generalization} In order to obtain Table~\ref{table: generalization method} from Table~\ref{table: Non-anonymized microdata table}, the composition of the three following A-COMPASS requirements need to be applied.

\begin{tcolorbox}[colback=gray!10,colframe=gray!10, boxrule=0.6pt, arc=0pt]
  \begin{tabular}{rl}
  \gG{EACH} & \gG{FILTER} Postal code > '21200' \gG{AND} Postal code < '21299' \\
        :  & Postal code = '212**' \\
        :  & \gG{REPLACE} Postal code \gG{WITH} \. '212**';
\end{tabular}
\end{tcolorbox}

\begin{tcolorbox}[colback=gray!10,colframe=gray!10, boxrule=0.6pt, arc=0pt]
\begin{tabular}{rl}
   \gG{EACH}&\gG{FILTER} Postal code > '21400' \gG{AND} Postal code < '21499'\\
         : & Postal code = '214**'\\
         : & \gG{REPLACE} Postal code \gG{WITH} \. '214**';
\end{tabular}
\end{tcolorbox}

\begin{tcolorbox}[colback=gray!10,colframe=gray!10, boxrule=0.6pt, arc=0pt]
\begin{tabular}{rl}
   \gG{EACH}&\gG{FILTER} Postal code > '21100' \gG{AND} Postal code < '21199'\\
         : & Postal code = '211**' \\
         : & \gG{REPLACE} Postal code \gG{WITH} \. '211**';
\end{tabular}
\end{tcolorbox}

The assertions in all three requirements now represent accessing particular records, and the conditions that we want these particular records to satisfy. We explain the syntax on the example of the second requirement.

In the second requirement, Postal code = '214**' represents the condition we want some records from our table to satisfy.  The command EACH FILTER means that the assertion affects only filtered records that satisfy the Boolean expression Postal code > '21400' \gG{AND} Postal code < '21499'. If the assertion is not satisfied, action REPLACE Postal Code WITH '214**' replaces the value of the attribute Postal Code with '214**' only in the records where Postal Code starts with 214.

NOTE: The same result can be achieved by the following A-COMPASS requirement. 

\begin{tcolorbox}[colback=gray!10,colframe=gray!10, boxrule=0.6pt, arc=0pt]
\begin{tabular}{rl}
   \gG{EACH}&\gG{RESULT} \\
         : & \gG{NOT} (Postal code = '21410')\\
         : & \gG{REPLACE} Postal code \gG{WITH} \. '214**' ;
\end{tabular}
\end{tcolorbox}

This approach gives the same result like the previous one since there is only one record in the table for which Postal code = '21410'. However, the first approach is more intuitive because one filters the records that are to be affected and there is no need to handle the negation of a condition.

\end{example}


\begin{example} [OUTLIERS in A-COMPASS] \label{example: outlier} Outliers can be very useful for statistical analysis, yet can also be misleading. Therefore, they require special attention from the perspective of statistical analysis. When it comes to anonymity analysis, they require even greater attention because they often contain highly revealing information. For example, the record with RECORD ID = 1 from Table~\ref{table: Non-anonymized microdata table} is an outlier record. In the case of identity disclosure, an adversary will associate a specific person with the value AEC=2200. The 2200 kWh for annual electricity consumption is too low even for one person living in a house, leading to the conclusion that very likely no one lives in this house. In order to preserve the privacy of the owner, the data curator might consider replacing the value AEC=2200 with some other value. Additionally, if the condition applies to multiple records in the table, the AEC value may be replaced in all of them.

Let us assume that the minimum value corresponding to the consumption of one person is 3000 kWh, and that all values of AEC less than 3000 are too revealing. If we apply the deterministic approach, with the help of the REPLACE action, these values can be replaced, for example, with the mean value for AEC. The following A-COMPASS requirement represents this approach:

\begin{tcolorbox}[colback=gray!10,colframe=gray!10, boxrule=0.6pt, arc=0pt]
\begin{tabular}{rl}
   \gG{EACH}&\gG{RESULT}\\
         : & AEC $\geq$ 3000 \\
         : & \gG{REPLACE} AEC \gG{WITH} \. 6200;
\end{tabular}
\end{tcolorbox}

Replacing the outlier value with the mean value is the standard approach used in statistical analysis. However, it significantly changes the distribution of the attribute. That is why one can consider replacing the outlier value with some other value. For example, the data curator can decide to use the probabilistic approach and replace the values less than 3000 with a random value from the range that corresponds to the range of AEC values for household of one person. This approach does not affect the attribute distribution as much as the previous approach, and it is close to the differential privacy approach.  The following A-COMPASS requirement represents the probabilistic approach:

\begin{tcolorbox}[colback=gray!10,colframe=gray!10, boxrule=0.6pt, arc=0pt]
\begin{tabular}{rl}
   \gG{EACH}&\gG{RESULT}\\
         : & AEC $\geq$ 3000 \\
         : & \gG{RANDOM} AEC 3000 4000;
\end{tabular}
\end{tcolorbox}

The action RANDOM attribute lowerBound upperBound replaces the values of an attribute with an integer value between the lowerBound and the upperBound. It can also be used for categorical attributes, if preprocessing of the table is done in a way to replace the value of a category with an integer value.
    
\end{example}


\begin{example}[k-anonymity in A-COMPASS] \label{example: k-anonymity} The A-COMPASS assertion for Table~\ref{table: k-anonymity and l-diversity} that corresponds to the k-anonymity condition is the following:


\begin{tcolorbox}[colback=gray!10,colframe=gray!10, boxrule=0.6pt, arc=0pt]
\begin{tabular}{rl}
   \gG{EACH PROCESS}&\gG{COUNT}(*) \gG{AS} Class Size \gG{GROUP BY} Age, Postal Code\\
         : & Class Size $\geq$ k;
\end{tabular}
\end{tcolorbox}

This is an example of the A-COMPASS processQuery where COUNT(*) is used for generalization. This query counts the number of records in each group defined by Age and Postal Code. The counts represent the values of a new attribute, Class Size. The Boolean condition Class Size $\geq$ k corresponds to the k-anonymity condition.
    
\end{example}


\begin{example}[l-diversity in A-COMPASS]\label{example: l-diversity} The A-COMPASS assertion for Table~\ref{table: k-anonymity and l-diversity} that corresponds to the l-diversity condition is the following:


\begin{tcolorbox}[colback=gray!10,colframe=gray!10, boxrule=0.6pt, arc=0pt]
\begin{tabular}{rl}
   \gG{EACH PROCESS}&\gG{COUNT DISTINCT} AEC \gG{AS} Diversity AEC \gG{GROUP BY} Age, Postal Code\\
         : & Diversity AEC $\geq$ l;
\end{tabular}
\end{tcolorbox}

This is an example of the A-COMPASS processQuery where COUNT DISTINCT is used for generalization. This query counts the number of different values of the attribute AEC in each group defined by Age and Postal Code. The counts represent the values of a new attribute, Diversity AEC. The Boolean condition Diversity AEC $\geq$ l corresponds to the l-diversity condition.
    
\end{example}


\begin{example}[Homogeneity attack in A-COMPASS] \label{example: homogeneity attack} As we have stated before, Table~\ref{table: k-anonymity and l-diversity} is vulnerable to the homogeneity attack since all the values for the sensitive attribute in the third class are the same. This can be prevented by the action REJECT, which will delete all the records of the vulnerable equivalence class.


\begin{tcolorbox}[colback=gray!10,colframe=gray!10, boxrule=0.6pt, arc=0pt]
\begin{tabular}{rl}
   \gG{EACH PROCESS}&\gG{COUNT DISTINCT} AEC \gG{AS} Diversity AEC \gG{GROUP BY} Age, Postal Code\\
         : & Diversity AEC > 1\\
         : & \gG{REJECT};
\end{tabular}
\end{tcolorbox}
    
\end{example}

Finally, by applying the previous A-COMPASS requirements to Table~\ref{table: Non-anonymized microdata table}, Table~\ref{table: final} is obtained.

\begin{table}[h!]
\centering

\begin{minipage}{0.34\linewidth}
    \centering
    
    \caption{Resulting table after A-COMPASS application}
\label{table: final}
   \begin{tabular}{cccc}
\toprule
Record ID & Age & Postal Code & AEC (kWh) \\
\midrule
1  & 54  & 212** & 3500 \\
2  & 54  & 212** & 7400 \\
3  & 54  & 212** & 8600 \\
\hline
4  & 80  & 214** & 10500 \\
5  & 80  & 214** & 3500 \\
6  & 80  & 214** & 8600 \\
\hline
9  & 45 & 211** & 6200 \\
10 & 45  & 211** & 5400 \\
\bottomrule
\end{tabular} 
\end{minipage}
\end{table}

\subsection{A-COMPASS vs. COMPASS} 

\vspace{5pt}

In the following, we explain the main differences between the COMPASS and the A-COMPASS language. 

\vspace{5pt}

\textbf{One Record - One Group vs. One Record - One Person.} Standard microdata tables are in the form of one record - one person, like Table~\ref{table: Non-anonymized microdata table}. However, COMPASS works only with pre-processed tables in the form of one record - one group. That means that the anonymization methods are already applied and a new attribute called Number is introduced. Number represents the number of persons in each group. COMPASS then allows reasoning about whether this pre-processed table satisfies some anonymity conditions. In addition to the tables in the form of one record - one group, A-COMPASS can also work with tables in the form of one record - one person. \\

\textbf{ZERO vs. REPLACE.} The first modification of the syntax is in the generalization of the action ZERO from the COMPASS language into a new action REPLACE, which gives more flexibility. ZERO replaces the values of an attribute with the value 0, in turn, REPLACE allows the user to replace the values of an attribute with an arbitrary value. This action also enables the suppression of an attribute, e.g., Example~\ref{example: suppression}. In addition, a composition of requirements with REPLACE enables the generalization of an attribute, e.g., Example~\ref{example: generalization}. \\

Furthermore, standard COMPASS language does not allow reasoning about $l$-diversity in the standard one record - one person setting. In order to do this in COMPASS, two steps are needed. First, duplicates must be collapsed, so each sensitive value appears only once per equivalence class. Then it can be counted how many records exist per group. In order to do so, sub-queries are needed, or a composition of two requirements. However, sub-queries are not defined. The composition of two queries is also not possible because there is no action in COMPASS that can collapse duplicates. That is why we introduce COUNT DISTINCT, a new aggregation operation in A-COMPASS.\\

\textbf{COUNT DISTINCT.} The second modification is in the extension of the COMPASS language with a new aggregation operation called COUNT DISTINCT, which counts the distinct values of an attribute. This aggregation operation allows reasoning about the l-diversity condition in the one record - one person setting, like in Example~\ref{example: l-diversity}.\\

\textbf{No JOIN.} The original COMPASS language supports JOIN action which allows users to join groups in one record - one group setting. Basically, the motivation behind the JOIN action is the generalization method. The action sums the values of the attribute Number for joined groups. In the A-COMPASS language, generalization can be performed by the composition of requirements with REPLACE action. That is why A-COMPASS does not contain (and need) the action JOIN.\\

To summarize, from the aspect of syntax, requirements from Example~\ref{example: outlier}  (RANDOM part) and Example~\ref{example: k-anonymity}, are examples of correct syntax in both COMPASS and A-COMPASS language, while requirements from Examples~\ref{example: suppression}, \ref{example: generalization}, \ref{example: outlier} (REPLACE part), \ref{example: l-diversity}, \ref{example: homogeneity attack} are correct only in A-COMPASS. Looking at the broader picture, from the aspect of the table type to which these requirements apply, all presented examples are correct only in A-COMPASS.

\subsection{A-COMPASS vs. SQL}

\vspace{5pt}

Both A-COMPASS and COMPASS are SQL based languages. Consequently, they inherit a large part of SQL features. In the following, we elaborate the main similarities and differences between A-COMPASS and SQL that are relevant for the problem.\\

\textbf{Sets vs. Bags. } The SQL semantics introduced in \cite{GuLi2017}, is defined over bags (multisets). The A-COMPASS semantics follows this principle. Although the first works on formal semantics of SQL were using set semantics \cite{MalechaMSW10, Negri1991}, it turned out that this approach is not expressive enough. Bag semantics is needed for calculations with duplicates. In the original table duplicate records are not allowed (due to the uniqueness of the primary key); however, when it comes to calculations with COUNT and SUM duplicate values can appear. It is the same in A-COMPASS. In fact, duplicates are the core of the anonymity analysis. Each record represents one person, while two identical records (if the primary key is not observed) represent two different persons with the same characteristics. The k-anonymity principle is based on this fact.\\

\textbf{Multiple Tables.} While SQL deals with multiple tables at a time, A-COMPASS is always interpreted on a single input table. That is why in A-COMPASS we do not have to deal with disambiguation of attributes with the same name across tables as in SQL. In A-COMPASS, attributes are elements of a set of attributes, and the only requirement is that there are no attributes with the same name.\\

\textbf{SELECT Queries.} In SQL,  \textsf{SELECT}...\textsf{FROM}...\textsf{WHERE} is the main query type. In A-COMPASS, there are no queries in this form. However, the semantics of requirements with the filterQuery is the same as semantics of the \textsf{FROM}...\textsf{WHERE}, which is a part of the \textsf{SELECT}...\textsf{FROM}... \textsf{WHERE} query restricted to one table.\\

\textbf{Subqueries.} SQL allows subqueries, such as \textsf{SELECT}...\textsf{FROM} ( \textsf{SELECT}...\textsf{FROM}...). They can also produce attributes with the same name. A-COMPASS does not support subqueries.\\

\textbf{Three-valued Logic.} SQL operates with three-valued logic, using three truth values: true (\textbf{t}), false (\textbf{f}) and unknown (\textbf{u}). The unknown is also called \textsf{NULL}. First, SQL supports missing values, and they are treated as \textsf{NULL}s. Second, \textsc{NULL} can also appear in conditions of a query or in queries with aggregate operations, such as SUM, MIN, MAX, COUNT, and COUNT DISTINCT. This happens when there are no records that satisfy the condition. COUNT and COUNT DISTINCT on an empty bag produce 0, while SUM, MIN, and MAX produce \textsf{NULL}. A-COMPASS operates with two-valued logic under certain justified restrictions. A-COMPASS is interpreted only on a complete table, without missing values. Therefore, the table with missing values is supposed to be preprocessed before the analysis in A-COMPASS. This restriction comes naturally since the released tables with microdata are usually complete. There is no point in releasing the incomplete record. When it comes to queries with aggregate operations, they can be defined on non-empty bags to avoid \textsf{NULL} values. Such results are not relevant for the anonymity analysis, so the language does not lose anything on it's expressiveness by using two-valued logic.\\

\textbf{Randomness.} Most SQL engines use a deterministic pseudo-random number generator. Randomness is exposed through functions such as \textsf{random()} or \textsf{RAND()} whose range is $[0,1)$. However, formal semantics of SQL that explicitly model these functions are rarely addressed in the literature. The denotational semantics of programming languages in \cite{BoDAetAl2016, DaSiSmi2023} provide a purely deterministic semantics for probabilistic programs, using deterministic pseudo-random samplers, called deterministic stream or random trace. A-COMPASS semantics follows the same idea. 

\section{Semantics}\label{semantics}
In this section we develop the denotational semantics for the A-COMPASS language. Our approach is inspired by Guagliardo and Libkin \cite{GuLi2017}, who introduce denotational semantics for SQL. First, we explain the data model and the notation we are going to use. Then we present the semantics.

\subsection{Data Model and Notation}


A-COMPASS data model takes into account two sets:
\begin{itemize}
    \item $N$ - a countably infinite set of attribute names (Names), and
    \item $D$ - a non-empty set (without \textsc{NULL}) of values that populate the data-base(Domain) .
\end{itemize}

The non-empty tuple of distinct names from $N$ is called a schema. The domain set $D$ comprises primitive domains, such as integers and strings. We assume that there is a well-defined semantics of predicates of base types. 

A {\bf record} $r$ is a tuple of elements from $D$. A {\bf relation} $R$ is a finite bag (multiset) of records \footnote{Bags will be denoted by curly braces unless explicitly stated otherwise.}. With $S=(A_1, A_2, \cdots, A_n)$ we denote the schema of the relation $R$, and with $r=(a_1, a_2, \cdots, a_n)$ a record on the schema $S$. \\

For the multiplicity of a record $r$ in a bag $R$, we will use the following notation:
\begin{itemize}
    \item $\#(r,R)$ for the multiplicity of $r$ in $R$;
    \item $r\in_k R$ for $\#(r,R)=k$;
    \item $r \in R$ to indicate that $r\in_k R$ for some $k>0$;
        \item $\displaystyle |R|=\sum_{r} \#(r,R)$ for the number of records in $R$ (the cardinality of $R$).
\end{itemize}

We will also use the following definitions of bag operations \cite{GruMi1996, LibWo1997}:

\begin{itemize}
    \item $\#(r, R_1 \cup R_2)=\#(r, R_1)+\#(r, R_2)$;
    \item $\#(r, R_1 \cap R_2)=\min (\#(r, R_1),\#(r, R_2))$;
    \item  $\#(r, R_1 - R_2)= \max (\#(r, R_1)-\#(r, R_2), 0)$.
\end{itemize}

\noindent For turning a bag into a set (which will be needed for COUNT DISTINCT) we are going to use {\bf duplicate elimination operation} $\upvarepsilon$ defined in \cite{GuLi2017} in the following way:

$$\#(r, \upvarepsilon(R))=\min(\#(r,R),1).$$

\noindent For actions REPLACE and RANDOM, we need to define \textbf{record update}.

\begin{definition}[Record Update]\label{definition: record update}
    For a record $r=(a_1, a_2, \cdots, a_n)$ on a schema $l(R)=(A_1, A_2, \cdots, A_n)$, an attribute $A_i\in l(R)$ and a value $a\in D$, the update of the record is defined by $$r[A_i \xrightarrow{} a ] = (a_1, a_2, \cdots, a_{i-1}, a, a_{i+1}, \cdots, a_n).$$
\end{definition}

\noindent For action RANDOM, we also need to define a \textbf{random trace}.

\begin{definition}[Random Trace] \label{definition: random trace}
    A random trace $\rho$ is an infinite sequence $\rho=(\rho_1, \rho_2, \cdots)\in [0,1)^\mathbb{N}$.
\end{definition}

\noindent A fixed random trace in the beginning (at the initial state) allows us to model RANDOM deterministically. For fixed $\rho \in [0,1)$ and integer bounds $l$ (lowerbound) and $u$ (upperbound), \textsf{RANDOM} returns the integer value $$random(l,u,\rho)=l+\lfloor \rho \cdot (u-l+1)\rfloor.$$

\noindent Therefore, the semantic state will be the pair $(R, \rho)$, where $R$ is a relation and $\rho$ is a random trace.\\

\noindent In order to obtain the semantics that is fully deterministic, we also need to define \textbf{canonical ordering} of a bag.\\

\begin{definition} [Canonical Ordering] \label{definition: canonical ordering}

\noindent For a bag $R$, the canonical ordering $order(R)=(r_1, \cdots, r_m)$ is a list of records in non-decreasing order, where each record appears as many times as it's multiplicity in $R$.
    
\end{definition}

\noindent For the non-decreasing order, we assume the natural order for numbers and the lexicographic order for strings.\\

\noindent Why do we need canonical ordering? The action RANDOM takes a record, takes one value from the random trace, and updates the record as described above. For the next record, it takes the next value from the random trace. However, bags do not have an order. When it comes to the semantics, different orderings of the same bag $R$ will lead to different outputs $(R,\rho)$ and the semantics will not be deterministic. This is why, at the beginning, we need to fix a random trace and the canonical ordering of the bag. 

\subsection{Semantics}


Our goal is to design a semantics of the requirements of the A-COMPASS language. The semantics of a requirement $r$ will be denoted by $\llbracket r \rrbracket$. This is a function that takes as input: a relation $R$(a database) and a random trace $\rho$. The output of the function will be the relation and the random trace obtained by executing the requirement $\llbracket r \rrbracket (R,\rho)$. The schema of  $\llbracket r \rrbracket (R,\rho)$ will be denoted by $l(R)$. It represents the tuple of attribute names associated with $\llbracket r \rrbracket (R,\rho)$ and it is related to the first part of the requirement, the assertion. The schema is defined in Figure~\ref{figure: A-COMPASS schema}, depending on the type of the assertion.

\begin{figure}[H]
    \centering
    \begin{minipage}{0.75\textwidth}
    \rule{\linewidth}{0.5pt}
    \begin{tabular}{@{}l@{}}
    $l(R)=S$, \text{tuple of attribute names for the original relation}\\[0.6em]
    $l(\text{RESULT})=l(R)$\\[0.6em]
    $l(\text{FILTER} \: \varphi)=l(R)$\\[0.6em]
    $l(\text{PROCESS} \: f(B) \:\text{as} \: C \:\text{WHERE} \:\psi \:\text{GROUP BY} \:(G_1, \cdots, G_m))=(G_1, \cdots, G_m, C)$\\[0.6em]
    $l(\text{PROCESS} \: f(B) \:\text{as} \: C \:\text{WHERE} \:\psi ) =(C)$\\
    \rule{\linewidth}{0.5pt}
    \end{tabular}
    \end{minipage}
    \caption{The schema of the A-COMPASS requirements}
    \label{figure: A-COMPASS schema}
\end{figure}

Now, we define a function that maps attribute names to values of a record.

\begin{definition}[Environment.] An environment is partial mapping from $N$, attribute names,  to $D$, values. Given a schema $S=(A_1,A_2,\cdots,A_n)$ and a record $r=(a_1,a_2, \cdots,a_n)$ on $S$, the environment is defined by

$$\eta_{S,r}(A_i)=a_i.$$
    
\end{definition}

In the following, we describe the semantics of the A-COMPASS language. The input of the semantic function $\llbracket \cdot \rrbracket$ depends on the syntactic category in the matter: for terms, the only input is the environment; for predicates and assertions, the inputs are the relation and the environment; finally, for actions, the inputs are the relation, environment, random trace and the affected bag \footnote{An affected bag by an assertion is not a syntactic category of the A-COMPASS language. However, it is a product of the assertion needed to perform an action.}. \\

\textbf{Terms.} A term is either a constant $c$ from $D$, or an attribute reference $A$ \footnote{We use the same notation for attribute names and attribute references. As an attribute name, $A$ denotes an attribute (column) name in a schema $S$. As an attribute reference, $A$ denotes the value of the attribute named $A$.}. The semantics of a term(and tuple of terms) is given by the environment $\eta$ in Figure~\ref{figure: A-COMPASS terms}.

\begin{center}

\hrule
\vspace{0.5em}

\begin{minipage}[t]{0.4\linewidth}
\centering 
\textbf{TERMS}
\[\llbracket t \rrbracket_\eta=\begin{cases*}
\eta(A)& \text{if} t=A\\
c& \text{if} t=c
\end{cases*}\]
\[\llbracket (t_1, t_2, \cdots, t_n) \rrbracket_\eta = (\llbracket t_1 \rrbracket_\eta, \llbracket t_2 \rrbracket_\eta, \cdots, \llbracket t_n\rrbracket_\eta)\]
\end{minipage}
\hfill
\begin{minipage}[t]{0.55\linewidth}
\centering 
\textbf{PREDICATES}
\[\llbracket A \, \theta \, c\rrbracket_\eta =
\begin{cases}
    \text{true} & \text{if } \llbracket A \rrbracket_\eta \, \theta \, \llbracket c\rrbracket_\eta = \eta(A) \, \theta \, c \text{ holds} \\
    \text{false} & \text{otherwise}
\end{cases}\]
\[\llbracket NOT \, \varphi \rrbracket_\eta = \neg \llbracket \varphi \rrbracket_\eta\]
\[\llbracket \varphi_1 \: AND \: \varphi_2 \rrbracket_\eta= \llbracket \varphi_1 \rrbracket_\eta \, \land \, \llbracket \varphi_2 \rrbracket_\eta\]
\[\llbracket \varphi_1 \: OR \: \varphi_2 \rrbracket_\eta= \llbracket \varphi_1 \rrbracket_\eta \, \lor \, \llbracket \varphi_2 \rrbracket_\eta\]
\end{minipage}

\vspace{0.5em}
\hrule

\vspace{0.5em}

\captionof{figure}{Semantics of the A-COMPASS terms and predicates}
\label{figure: A-COMPASS terms}
\end{center}

\textbf{Predicates.} Predicates correspond to A-COMPASS Boolean conditions(boolCond), Boolean literals(boolLiteral), Boolean products(boolProduct), and Boolean expressions(boolExpr). Their evaluation is two-valued $\llbracket P \rrbracket_\eta \in \{true, false\}$. We distinguish the following cases:\\

\begin{itemize}

\item \underline{boolCond} Boolean conditions are atomic predicates in a form $A \,\theta \, c$, where $A$ is an attribute reference, $\theta \in \{<,>,=,\leq, \geq\}$ is a comparison operator, $c$ is a constant; \\

\item \underline{boolLiteral} Boolean literal is a negation of a Boolean condition; \\

\item \underline{boolProduct and boolExpr} Boolean product and Boolean expression represent the conjuction and the disjunction of Boolean conditions, respectively. \\

\end{itemize}

The semantics of A-COMPASS predicates is given in Figure~\ref{figure: A-COMPASS terms}.\\

\subsubsection{Assertions.} 

\vspace{5pt}

The syntax of the A-COMPASS assertion is given by:\\

\begin{tabular}{lcl}
\gG{assertion} & ::= & ( 'EACH' \gG{result} ':' \gG{boolExpr} ) \\
&  & \textbar( 'SOME' \gG{result} ':' \gG{boolExpr} )
\end{tabular}

\vspace{5pt}

In order to define the semantics of an assertion, we first need to define the semantics of a result. \\

\textbf{Result.} The semantics of a result is given in Figure~\ref{figure: result semantics}. It represents the whole relation (if the result is the keyword RESULT), or the semantics of a resultQuery (in all other cases). \\

\begin{figure}[h!]
    \centering
    \begin{minipage}{0.65\textwidth}
    \rule{\linewidth}{0.5pt}
    \vspace{1em}
    $$
\llbracket \text{result} \rrbracket (R)=
\begin{cases}
    R & \text{if } \text{result}=\text{'RESULT'}\\
    \llbracket \text{resultQuery} \rrbracket (R) & \text{otherwise}
\end{cases}
$$

    \vspace{1em}
    \rule{\linewidth}{0.5pt}
    \end{minipage}
    \caption{Semantics of the A-COMPASS result}
    \label{figure: result semantics}
\end{figure}

Further, we need the semantics of resultQuery. Since resultQuery is defined by either a filterQuery or a processQuery, we will first give the semantics of a filterQuery.\\

\textbf{filterQuery.} The semantics of a filter query is defined in Figure~\ref{figure: filter semantics}. It corresponds to the FROM... WHERE part of the SELECT... FROM... WHERE query from SQL, and it represents the sub-bag of $R$ composed of records that meet the defined condition. The multiplicity of a record in the sub-bag is the same as the multiplicity of the same record in $R$. \\

\begin{figure}[h!]
    \centering
    \begin{minipage}{0.65\textwidth}
    \rule{\linewidth}{0.5pt}
    \vspace{1em}
    $$\llbracket \text{FILTER} \, \varphi \rrbracket = \left\{
\underbrace{r, \ldots, r}_{k \text{ times}}
\;\middle|\;
r \in_k R,\;
\llbracket \varphi \rrbracket_\eta = true
\right\}$$

    \vspace{1em}
    \rule{\linewidth}{0.5pt}
    \end{minipage}
    \caption{Semantics of the A-COMPASS filter query}
    \label{figure: filter semantics}
\end{figure}

For the aggregation in processQuery, we will use the following definition of aggregate functions on bags.

\vspace{10pt}

\begin{definition}[Aggregate Functions] For a bag $R$ on a schema $S$ and an attribute $A\in S$, summation, maximization, minimization, counting, and distinct counting are defined by:

\begin{itemize}

\item $SUM_A(R)=\sum_{r\in R} \#(r,R)\cdot\llbracket A \rrbracket_{\eta_{S,r}};$

\item $MAX_A(R)=\max \{\llbracket A \rrbracket_{\eta_{S,r}}, r\in R\};$

\item $MIN_A(R)=\min \{\llbracket A \rrbracket_{\eta_{S,r}}, r\in R\};$

\item $COUNT_*(R)=|R|;$

\item $COUNT \; DISTINCT_A(R)=|\upvarepsilon(\{\llbracket A\rrbracket_{\eta_{S,r}}, r \in R \})|.$

\end{itemize}
   
\end{definition}

\textbf{processQuery. } Let PROCESS $f_B$ AS $C$ [WHERE $\psi$][GROUP BY $G$] be a process query where $f$ is an aggregate function applied to the attribute $B$ (or *), $C$ is the name of the introduced attribute, $\psi$ is a Boolean expression used for filtering, $G=(G_1, \cdots, G_m)$ where $G_i \in  S$ is a set of attributes from a schema $S$ used for grouping. We will introduce the semantics of this query in three steps.

\begin{enumerate}
    \item FILTERING. If WHERE $\psi$ is present, the query first filters some records and produces a sub-bag of $R$. The filtered bag $R'$ is given by
    \[R'=\left\{
\underbrace{r, \ldots, r}_{k \text{ times}}
\;\middle|\;
r \in_k R,\;
\llbracket \psi \rrbracket_{\eta_{S,r}} = true
\right\}.\]

If there is no filtering, $R'=R$.

\item GROUPING. If GROUP BY $G$ is present, the query then partitions the records of the filtered bag into groups. Each group will have it's own group key. Let $K$ denote a set of distinct group keys that occur in $R'$. $K$ is then given by:
\[K=\{\ \llbracket  G \rrbracket_{\eta_{S,r}}, r \in R'\}, \text{where}\]
\[\llbracket  G \rrbracket_{\eta_{S,r}}=(\llbracket  G_1 \rrbracket_{\eta_{S,r}},\cdots, \llbracket  G_m \rrbracket_{\eta_{S,r}} ).\]

For a group with a group key $g \in K$, the sub-bag of records that constitutes this group is given by:

\[G_g=\left\{
\underbrace{r, \ldots, r}_{k \text{ times}}
\;\middle|\;
r \in_k R',\;
\llbracket G \rrbracket_{\eta_{S,r}} = g
\right\}.\]

\item AGGREGATION. In the third step, the aggregate values per group $f_B(G_g)$ are computed. If the grouping step is skipped, only one aggregate value $f(R')$ is computed for the entire filtered bag.  The function $f$ can be one of the aggregate functions we have defined above (applied to the attribute $B$ and the corresponding sub-bag).

\end{enumerate}

Finally, the semantics of a processQuery is defined in Figure~\ref{figure: process semantics}. \\

\begin{figure}[h!]
    \centering
    \begin{minipage}{0.65\textwidth}
    \rule{\linewidth}{0.5pt}
    \vspace{1em}
   $$\llbracket \text{PROCESS} \, f_B \, \text{AS} \: C \: \text{WHERE} \: \psi \: \text{GROUP BY} \, G\rrbracket = \{(g, f_B(G_g)), g\in K\}$$

$$\llbracket \text{PROCESS} \, f_B \, \text{AS} \: C \: \text{WHERE} \: \psi \rrbracket = \{(f(R'))\}$$

    \vspace{1em}
    \rule{\linewidth}{0.5pt}
    \end{minipage}
    \caption{Semantics of the A-COMPASS process query}
    \label{figure: process semantics}
\end{figure}

When grouping is present, the multiplicity of each tuple $(g, f_B(G_g))$ is one, since the group key is unique per group. If there is no grouping, there is only one tuple $(f(R'))$. \\

Finally, we define the semantics of an assertion. Let $P=\llbracket result \rrbracket (R)$ with schema $S'=l(result)$ (the result schema is precisely defined in Figure~\ref{figure: A-COMPASS schema}). The semantics of an assertion is then defined in Figure~\ref{figure: assertion semantics}. \\

\begin{figure}[h!]
    \centering
    \begin{minipage}{0.65\textwidth}
    \rule{\linewidth}{0.5pt}
    \vspace{1em}
   $$\llbracket \text{EACH} \: \text{result} \: \varphi \rrbracket (R)=true \iff \forall r \in P: \llbracket \varphi\rrbracket_{\eta_{S',r}}= true$$

$$\llbracket \text{SOME} \: \text{result} \: \varphi \rrbracket (R)=true \iff \exists r \in P: \llbracket \varphi\rrbracket_{\eta_{S',r}}=true$$

    \vspace{1em}
    \rule{\linewidth}{0.5pt}
    \end{minipage}
    \caption{Semantics of the A-COMPASS assertion}
    \label{figure: assertion semantics}
\end{figure}

\subsubsection{Affected Bag.}

For each assertion there will be a bag of records which violates the assertion and needs to be altered by the following action. This bag will be called an \textit{affected bag} and it is a sub-bag of the original relation $R$. If there are no records that violate the assertion, the affected bag will be an empty bag.

For an assertion of the type 'SOME' \gG{result} ':' \gG{boolExpr}, the affected bag will be the whole relation $R$. When an assertion with 'SOME' is violated, it means that there are no records in the whole relation that satisfy the required condition. 

For an assertion of the type 'EACH' \gG{result} ':' \gG{boolExpr}, we need to differentiate between several sub-types of assertions in order to define the affected bag.\\

\begin{definition}[Affected Bag for EACH RESULT : \gG{boolExpr}] \label{definition: aff bag for EACH RESULT} For a relation $R$ and an assertion EACH RESULT : $\varphi$, where $\varphi$ is a Boolean expression, the affected bag is defined by:

\[\mathrm{AFF}(R, \text{EACH RESULT :} \, \varphi) = \left\{
\underbrace{r, \ldots, r}_{k \text{ times}}
\;\middle|\;
r \in_k R,\;
\llbracket \varphi \rrbracket_{\eta_{S,r}} = false
\right\}.\]
    
\end{definition}

\begin{definition}[Affected Bag for EACH FILTER \gG{boolExpr} : \gG{boolExpr}] \label{definition: aff bag for EACH FILTER} For a relation $R$ and an assertion EACH FILTER $\psi$ : $\varphi$, where $\psi$ and $\varphi$ are Boolean expressions, the affected bag is defined by:

\[\mathrm{AFF}(R, \text{EACH FILTER } \psi \, \text{:} \, \varphi) = \left\{
\underbrace{r, \ldots, r}_{k \text{ times}}
\;\middle|\;
r \in_k R,\; \llbracket \psi \rrbracket_{\eta_{S,r}} = true \land
\llbracket \varphi \rrbracket_{\eta_{S,r}} = false
\right\}.
\]
    
\end{definition}

\begin{definition}[Affected Bag for EACH PROCESS with GROUP BY ] \label{definition: aff bag for EACH PROCESS GROUP BY} For a relation $R$ and an assertion EACH PROCESS $f_B$ AS $C$ WHERE $\psi$ GROUP BY $G$ : $\varphi$, where $f$ is an aggregate function applied to the attribute $B$, $C$ is the name of the introduced attribute, $G=(G_1,\cdots, G_m)$ is a set of attributes from a schema $S$ used for grouping, $\psi$  and $\varphi$ are Boolean expressions:

\[\mathrm{AFF}(R, \text{EACH PROCESS} \; f_B \:  \text{AS} \; C \; \text{WHERE} \; \psi \; \text{GROUP BY}\;  G \, \text{:} \, \varphi) = \left\{
\underbrace{r, \ldots, r}_{k \text{ times}}
\;\middle|\;
r \in_k R',\; \llbracket G \rrbracket_{\eta_{S,r}} \in K_v
\right\},
\]

where $[K_v=\{g \in K \; | \; \llbracket \varphi \rrbracket_{\eta_{S', r_a}}=false\}$ and $S'=(G_1, \cdots, G_m, C) \; , \; r_a=(g, f_B(Gg)).$
    
\end{definition}

Thus, for the affected bag of an assertion EACH PROCESS with FILTER and GROUP BY, we first need to identify a set of violating group keys $K_v$. The assertion condition $\varphi$ is evaluated on a schema of the PROCESS query $S'$ and an aggregated record $r_a$ from the semantics of the PROCESS query. To be more specific, we identify groups whose aggregate values do not meet the condition. Then, we identify records from the filtered bag $R'$ that constitute these groups.\\

\begin{definition}[Affected Bag for EACH PROCESS without GROUP BY] \label{definition: aff bag for EACH PROCESS no GROUP BY} For a relation $R$ and an assertion EACH PROCESS $f_B$ AS $C$ WHERE $\psi$ : $\varphi$, where $f$ is an aggregate function applied to the attribute $B$, $C$ is the name of the introduced attribute, and $\varphi$ is a Boolean expression:

\begin{itemize}
    \item IF $\llbracket \text{PROCESS} \; f_B \:  \text{AS} \; C \; \text{WHERE} \; \psi \rrbracket_R \neq \emptyset$ and $\llbracket \varphi \rrbracket_{\eta_{S', r_a}}=false$,\\ $\mathrm{AFF}(R, \text{EACH PROCESS} \; f_B \:  \text{AS} \; C \; \text{WHERE} \; \psi \, \text{:} \, \varphi)=R'$;
    \item Otherwise, $\mathrm{AFF}(R, \text{EACH PROCESS} \; f_B \:  \text{AS} \; C \; \text{WHERE} \; \psi \, \text{:} \, \varphi)=\emptyset$.
\end{itemize}

Here, $S'=(C)$ is a schema of the PROCESS query without GROUP BY, and $r_a=(f_B(R'))$ is a single aggregated record from the semantics of the PROCESS query.

\end{definition}

\subsubsection{Actions.}

An action is executed if an assertion is violated. Actions transform the semantic state $(R, \rho)$ in the way described below. \\

\textbf{REJECT action.} The action REJECT removes all affected records from the table. If $V$ is the affected bag for the required assertion, the semantics of the REJECT action is then given in Figure~\ref{figure: reject semantics}.\\

\begin{figure}[h!]
    \centering
    \begin{minipage}{0.5\textwidth}
    \rule{\linewidth}{0.5pt}
   
$$\llbracket \text{REJECT} \rrbracket (R,\rho, V) = (R-V, \rho)
$$
    \rule{\linewidth}{0.5pt}
    \end{minipage}
    \caption{Semantics of the action REJECT}
    \label{figure: reject semantics}
\end{figure}

We have already stated that for an assertion with SOME, the affected bag will be the whole relation $V=R$. In that case, the REJECT action will delete the entire table and $\llbracket \text{REJECT} \rrbracket (R,\rho, v) = (\emptyset, \rho)$.\\

\textbf{REPLACE action.} The action REPLACE A WITH \gG{a}  replaces the values of an attribute A with the value \gG{a} in all affected records. Let $V$ be the affected bag for the required assertion, the semantics of the action REPLACE A WITH \gG{a} action is then given in Figure~\ref{figure: replace semantics}.\\

\begin{figure}[h!]
    \centering
    \begin{minipage}{0.6\textwidth}
    \rule{\linewidth}{0.5pt}
   
$$\llbracket \text{REPLACE  A  WITH  a }\: \rrbracket (R,\rho, V) = ((R-V)\cup V_a, \rho)
$$

$$\text{where}$$

$$V_a = \left\{
\underbrace{r[A \xrightarrow{} a ], \ldots, r[A \xrightarrow{} a ]}_{k \text{ times}}
\;\middle|\;
r \in_k V
\right\}
$$
    \rule{\linewidth}{0.5pt}
    \end{minipage}
    \caption{Semantics of the action REPLACE}
    \label{figure: replace semantics}
\end{figure}

\textbf{RANDOM action.} The action RANDOM A l u replaces the values of an attribute A with the random integer value between the lower bound (l) and the upper bound (u) in all affected records. Let $V$, $|V|=m$ be the affected bag and $\text{order}(V)=(r_1, \cdots, r_m)$ be it's canonical ordering. Since $|V|=m$ the first $m$ values of the random trace $\rho$ will be consumed (one per record $r_i$ from the affected bag), and that is why we will decompose $\rho$ in $\text{pref}_m(\rho) = (\rho_1, \cdots, \rho_m)$ and $\text{tail}_m(\rho) = (\rho_{m+1}, \rho_{m+2},\cdots)$. The semantics of the action RANDOM is then given in Figure~\ref{figure: random semantics}.\\

\begin{figure}[h!]
    \centering
    \begin{minipage}{0.6\textwidth}
    \rule{\linewidth}{0.5pt}
   $$\llbracket \text{RANDOM  A  l u }\: \rrbracket (R,\rho, V) = ((R-V)\cup V_r, \text{tail}_m(\rho))
$$ 

$$\text{where}$$ 

$$V_r=
\left\{
r_i\!\left[A\mapsto \text{random}(l,u,r_i)\right]
\ \middle|\ 1\le i\le m
\right\}
$$

    \rule{\linewidth}{0.5pt}
    \end{minipage}
    \caption{Semantics of the action RANDOM}
    \label{figure: random semantics}
\end{figure}

\subsubsection{Requirements.} 

\vspace{5pt}

Finally, A-COMPASS requirement $r$ is in the form assertion : action. Let $V$ be the affected bag for the required assertion, the semantics for a requirement is then given in Figure~\ref{figure: requirement semantics}. \\

\begin{figure}[h!]
    \centering
    \begin{minipage}{0.6\textwidth}
    \rule{\linewidth}{0.5pt}
  $$ \llbracket r \rrbracket (R, \rho) = \begin{cases}
(R, \rho) & \text{if} \; \llbracket  \text{assertion} \rrbracket = \text{true}\\
\llbracket \text{action} \rrbracket (R, \rho, V) & \text{otherwise}
\end{cases} $$

    \rule{\linewidth}{0.5pt}
    \end{minipage}
    \caption{Semantics of the A-COMPASS requirement}
    \label{figure: requirement semantics}
\end{figure}

\section{Properties of the Semantics}\label{properties}
In this section, we prove the fundamental properties of the introduced semantics. We then validate the anononymity properties of A-COMPASS language using this semantics.
\subsection{Fundamental Properties}

Here we prove the determinism and the compositionality of the given semantics.\\

Since we modeled RANDOM deterministically, it is essential to show that the output of the semantic function is uniquely determined for every requirement and every  initial semantic state.

\begin{thm}[Determinism]\label{theorem: determinism} For every requirement $r$ and every semantic state $(R, \rho) \in \Sigma$, where $\Sigma$ is a set of semantic states, there exists a unique state $(R', \rho') \in \Sigma$ such that

$$\llbracket r \rrbracket (R, \rho) = (R', \rho').$$
    
\end{thm}

\begin{proof}
    The semantics of A-COMPASS requirement is defined by:
$$ \llbracket r \rrbracket (R, \rho) = \begin{cases}
(R, \rho) & \text{if} \; \llbracket  \text{assertion} \rrbracket = \text{true}\\
\llbracket \text{action} \rrbracket (R, \rho, V) & \text{otherwise}
\end{cases}.$$

That is why we have to distinguish between two cases: $\llbracket  \text{assertion} \rrbracket = \text{true}$ and $\llbracket  \text{assertion} \rrbracket = \text{false}$:

\begin{itemize}
    \item $\llbracket  \text{assertion} \rrbracket = \text{true}$, then $\llbracket r\rrbracket (R, \rho)=(R, \rho)$ is uniquely determined;
    \item $\llbracket  \text{assertion} \rrbracket = \text{false}$, then
    $\llbracket r\rrbracket (R, \rho)=\llbracket \text{action} \rrbracket (R, \rho, V). $
    
    By the definition of an affected bag, $V$ is uniquely determined for a fixed assertion. Moreover, if:\\
      - action = REJECT, then $\llbracket \text{REJECT} \rrbracket (R, \rho, V) =(R-V, \rho)$ is uniquely determined;\\
      - action = REPLACE A WITH a, then $\llbracket \text{REPLACE A WITH a} \rrbracket (R, \rho, V) =((R-V) \cup V_a, \rho)$ is uniquely determined, because of the definition of $V_a$;\\
      - action = RANDOM A l u, then $\llbracket \text{RANDOM A l u} \rrbracket (R, \rho, V) =((R-V) \cup V_r, \text{tail}_m(\rho))$ is uniquely determined, because of the canonical ordering of R and V and the decomposition of $\rho = (\text{pref}_m(\rho), \text{tail}_m(\rho))$, which are fixed.

Therefore, $\llbracket r \rrbracket (R, \rho)$ is uniquely determined in all cases.

\end{itemize}
\end{proof}

Since A-COMPASS program is defined by sequential composition, it is essential to show what is the output of that program.

\begin{thm}[Sequential composition]\label{theorem: composition} Let $r_1, r_2, \cdots, r_n $ be a non-empty finite sequence of A-COMPASS requirements. Then $ \forall (R, \rho) \in \Sigma$, where $\Sigma$ is a set of semantic states, it holds:

$$\llbracket r_1 ; r_2 ; \cdots ; r_n\rrbracket (R, \rho) = \left(\llbracket r_n \rrbracket \circ \llbracket r_{n-1} \rrbracket \circ \cdots \circ \llbracket r_1 \rrbracket \right)(R, \rho).$$
    
\end{thm}

\begin{proof} We prove the statement by the induction on sequence length $n$.

\begin{itemize}
    \item \underline{Base case (n=1)} $\llbracket r_1 \rrbracket (R,\rho) = \llbracket r_1 \rrbracket (R,\rho)$ holds for all $(R,\rho) \in \Sigma.$
    \item \underline{Induction hypothesis (n=k)} Asssume that for some $k>1$, it holds $$ \displaystyle \llbracket r_1 ; r_2 ; \cdots ; r_k\rrbracket (R, \rho) = \left(\llbracket r_k \rrbracket \circ \llbracket r_{k-1} \rrbracket \circ \cdots \circ \llbracket r_1 \rrbracket \right)(R, \rho).$$
    \item \underline{Induction step (n=k+1)} Let $P_{k+1}=r_1; r_2; \cdots ; r_{k+1}$ and $P_{k}=r_1; r_2; \cdots ; r_{k}$, which is of the length $n$. Then $P_{k+1}=P_k; r_{k+1}$ and

    $$ \displaystyle \llbracket P_{k+1}\rrbracket (R,\rho) = \llbracket P_k; r_{k+1}\rrbracket (R,\rho)=\left(\llbracket r_{k+1}\rrbracket \circ \llbracket P_k \rrbracket \right)(R,\rho).$$

    By the induction hypothesis: $ \displaystyle \llbracket P_{k+1}\rrbracket (R,\rho) = \llbracket r_{k+1}\rrbracket \circ \left( \llbracket r_k \rrbracket \circ \cdots \circ \llbracket r_1\rrbracket \right)(R,\rho)$.\\
    By the asociativity of function composition, it follows $$ \displaystyle \llbracket P_{k+1}\rrbracket (R,\rho) = \left(\llbracket r_{k+1}\rrbracket \circ \llbracket r_k \rrbracket \circ \cdots \circ \llbracket r_1\rrbracket \right)(R,\rho).$$
\end{itemize}
    
\end{proof}

\subsection{Anonymity Properties}

Here we validate the following anonymity properties of A-COMPASS language: suppression soundness, generalization soundness, RANDOM soundness with respect to outlier treatment, k-anonymity soundness and completeness, and l-diversity soundness and completeness.\\

\noindent First we give definitions of anonymity properties expressed in A-COMPASS semantics.

\begin{definition}[Indistinguishability] \label{definition: indistinguishability} Two records $r_1, r_2 \in R$ are indistinguishable with respect to $A\in S$, where $S=l(R)$, if

$$ r_1 \sim_{A,R} r_2 \Longleftrightarrow \llbracket A \rrbracket_{\eta_{S,r_1}}=\llbracket A \rrbracket_{\eta_{S,r_2}}.$$
    
\end{definition}

\noindent Indistinguishability among records can also be defined with respect to multiple attributes ($S' \subseteq S$) in a similar way:

$$r_1 \sim_{S',R} r_2 \Longleftrightarrow \llbracket S' \rrbracket_{\eta_{S,r_1}}=\llbracket S'\rrbracket_{\eta_{S,r_2}}.$$

\begin{definition}[Equivalence Class] \label{definition: equivalence class} Let $QI=(Q_1, Q_2, \cdots, Q_m) \subset S$ be a tuple of quasi-identifiers for a relation $R$. The equivalence class of a record $r \in R$ with respect to $QI$ is the sub-bag of $R$ defined by

$$ [r]_{QI,R}=\left\{
\underbrace{\overline{r}, \ldots, \overline{r} }_{k \text{ times}}
\;\middle|\;
\overline{r} \in_k R,\;
\overline{r} \sim_{QI,R} r
\right\}. $$
    
\end{definition}

\begin{definition}[k-anonymity by A-COMPASS semantics] \label{definition: k-anonymity semantical}  Let $QI=(Q_1, Q_2, \cdots, Q_m) \subset S$ be a tuple of quasi-identifiers for a relation $R$. The relation $R$ satisfies the k-anonymity principle if and only if

$$\forall r\in R,  \; |[r]_{QI,R}|\geq k.$$
    
\end{definition}

\begin{definition}[l-diversity by A-COMPASS semantics]\label{definition: l-diversity semantical} Let $QI=(Q_1, Q_2, \cdots, Q_m) \subset S$ be a tuple of quasi-identifiers for a relation $R$. Let $A \subset S$ be the sensitive attribute for the same relation. The relation $R$ satisfies the l-diversity principle if and only if

$$\forall r\in R,  \; |\upvarepsilon (\{\llbracket A \rrbracket_{\eta_{S,r}},r \in [r]_{QI,R}\})|\geq l.$$
    
\end{definition}

\begin{thm}[Suppression soundness] \label{theorem: replace suppression} For any $R$, let $\llbracket \text{REPLACE} \; A \; \text{WITH} \; a \rrbracket (R, \rho, V) = (R',\rho)$ where $R' = (R-V) \cup V_a$, $V$ is the affected bag for a certain assertion, while $V_a$ is the corresponding replacing bag for the action REPLACE $A$ WITH $a$. Then $$\forall r_1, r_2 \in V, r'_1, r'_2 \in [r]_{A=a,R'}$$ where $r'_1$ and $r'_2$ are the transformed records of $r_1$ and $r_2$ respectively, and $[r]_{A=a,R'}$ is the equivalence class in $R'$ that corresponds to $A=a$.
    
\end{thm}

\begin{proof} Let us choose two arbitrary records from the affected bag, $r_1, r_2 \in V$. By the semantics of the action REPLACE $A$ WITH $a$,  $r'_1, r'_2 \in V_a$ and

$$\llbracket A \rrbracket _{\eta_{S,r'_1}} = a, \llbracket A \rrbracket _{\eta_{S,r'_2}} = a.$$

Since

$$a= \llbracket A \rrbracket _{\eta_{S,r'_1}} = \llbracket A \rrbracket _{\eta_{S,r'_2}} \Longleftrightarrow r'_1 \sim_{A,R'} r'_2$$

meaning that $r'_1$ and $r'_2$ are indistinguishable with respect to $A$.\\

Finally, by the definition of an equivalence class $r'_1, r'_2 \in [r]_{A=a,R'}.$
    
\end{proof}

\textbf{Interpretation of Theorem~\ref{theorem: replace suppression}.} Action REPLACE $A$ WITH $a$ groups all affected records into one equivalence class with respect to the attribute $A$, and thus enables suppression. As a result, REPLACE is sound with respect to suppression.

\begin{thm}[Replace action enforces indistinguishability] \label{theorem: replace indistinguishability} For any $R$, let $\llbracket \text{REPLACE} \; A \; \text{WITH} \; a \rrbracket (R, \rho, V) = (R',\rho)$, where $R' = (R-V) \cup V_a$, $V$ is the affected bag for a certain assertion, while $V_a$ is the corresponding replacing bag for the action REPLACE $A$ WITH $a$. Let $[r]_{A=a,R}$ and $[r]_{A=a,R'}$ be equivalence classes that corresponds to $A=a$ in $R$ and $R'$ respectively. Then 

$$|[r]_{A=a,R'}|\geq |[r]_{A=a,R}|.$$
    
\end{thm}

\begin{proof}

Since $R' = (R-V) \cup V_a$, we have

\begin{equation*}   
\begin{split}
|[r]_{A=a,R'}|
&=|[r]_{A=a,R-V}|+|[r]_{A=a,V_a}|\\
&=|[r]_{A=a,R}|-|[r]_{A=a,V}|+|[r]_{A=a,V_a}|.
\end{split}    
\end{equation*}

From $|V_a|=|V|$ and $[r]_{A=a,V_a}=V_a$, we have $|[r]_{A=a,V_a}|=|V|.$\\

Therefore, $|[r]_{A=a,R'}|=|[r]_{A=a,R}| - |[r]_{A=a,V}|+|V|$. Considering that $[r]_{A=a,V}$ is a sub-bag of $V$, $|[r]_{A=a,V}|\leq |V|$ and $|V| - |[r]_{A=a,V}| \geq 0 $, leading to $|[r]_{A=a,R'}| \geq |[r]_{A=a,R}|$.

\end{proof}

\textbf{Interpretation of Theorem~\ref{theorem: replace indistinguishability}.} Suppression with action REPLACE $A$ WITH $a$ enforces indistinguishability among records. After this action, there are more records that are indistinguishable with respect to $A$ (and value $a$) than it was before.

\begin{thm}[Generalization soundness] \label{theorem: generalization} Let $Q \in S$ be a quasi-identifier for a relation $R$, and $S$ be a schema for $R$. Let $g_1, g_2, \cdots, g_n$ be a non-empty sequence of generalization requirements, where $g_i ::= \text{EACH RECORD} \; \varphi_i: \; \text{REPLACE} \; Q \;\text{WITH} \; q_i$ and $q_i \in \{ q_1, q_2, \cdots, q_n \}$ is a generalization value from $D$. Let $R_0, R_1, \cdots, R_n$ be a non-empty finite sequence of relation states defined by: $R_0=R$, $(R_i, \rho_i)=\llbracket g_i \rrbracket(R_{i-1}, \rho_{i-1})$. Let $V_1, V_2, \cdots, V_n$ be a non-empty finite sequence of affected bags for requirements $g_1, g_2, \cdots, g_n$. Assume:
\begin{itemize}
    \item[i)] $V_i \subseteq R_{i-1}$ for all $i \in \{ 1, 2, \cdots, n \}$;
    \item[ii)] $V_i=V_i^0$ for all $i \in \{ 1, 2, \cdots, n \}$, where $V_i^0$ is the affected bag for $R=R_0$;
    \item[iii)] $V_i \cap V_j = \emptyset$ for all $i \neq j$;
    \item[iv)] $ \displaystyle \bigcup_{i=1}^n V_i =R.$
\end{itemize}
   Then, $\forall r \in R_n$, it holds $\llbracket Q \rrbracket_{\eta_{S,r}}\in \{ q_1, q_2, \cdots, q_n \}.$ 
\end{thm}

\begin{proof} By iii) and iv) it follows that each record $r \in R$ belongs to exactly one $V_i$, $i \in \{ 1, 2, \cdots, n \}$, so it is affected by only one qeneralization requirement $g_i$, $i \in \{ 1, 2, \cdots, n \}$.\\

We prove the statement by the induction on $n$.

\begin{itemize}
    \item \underline{Base case (n=1)} We show that $\forall r \in R_1$, $\llbracket Q \rrbracket_{\eta_{S,r}} \in \{ q_1 \}.$ By the definition of $R_1$ and i) we have $(R_1, \rho_1)=\llbracket g_1 \rrbracket (R_0, \rho_0)$ and $R_1=(R_0 - V_1) \cup V_{q_1}$, where $V_{q_1}$ is the replacing bag of $V_1$. From iv) we have: $V_1=R=R_0$. Therefore, $r \in V_1$ and $\llbracket Q \rrbracket_{\eta_{S,r}} =q_1 \in \{ q_1 \}.$
    \item \underline{Induction hypothesis (n=k)} Assume that $\forall r \in R_k$, it holds $\llbracket Q \rrbracket_{\eta_{S,r}} \in \{ q_1, q_2, \cdots, q_k \}.$
    \item \underline{Induction step (n=k+1)}. We show that $\forall r \in R_{k+1}$, $\llbracket Q \rrbracket_{\eta_{S,r}} \in \{ q_1, q_2, \cdots, q_{k+1} \}.$ By the definition of $R_{k+1}$ and the Theorem~\ref{theorem: composition} of sequential composition, we have: $(R_{k+1}, \rho_{k+1})=\llbracket g_{k+1} \rrbracket \left(R_k,\rho_k \right).$ From i) we have $V_{k+1} \subseteq R_k$ and $R_{k+1}=(R_k-V_{k+1}) \cup V_{q_{k+1}}$, where $V_{q_{k+1}}$ is the replacing bag for $V_{k+1}$. Let us choose an arbitrary $r \in R_{k+1}$. Therefore, either $r \in (R_k - V_{k+1})$, or $r \in V_{q_{k+1}}$: \\

    - if $r \in (R_k - V_{k+1})$, then by the induction hypothesis $\llbracket Q \rrbracket_{\eta_{S,r}}\in \{ q_1, q_2, \cdots, q_k \} \subset \{ q_1, q_2, \cdots, q_{k+1} \} $;\\
    - if $r \in V_{q_{k+1}}$, then $\llbracket Q \rrbracket_{\eta_{S,r}} = q_{k+1} \in \{ q_1, q_2, \cdots, q_{k+1} \}$.

\end{itemize}

\end{proof}

\textbf{Interpretation of Theorem~\ref{theorem: generalization}.} After $n$ sequential generalization requirements in which every record is affected by only one requirement and all records are affected by some requirement, i.e., affected bags are disjoint and their union is the whole relation, values of the attribute $Q$ will be replaced by a more general value in all records. There will be no record in which the value of $Q$ will differ from one of the predetermined generalized values. Thus, composition of requirements with action REPLACE is sound with respect to generalization.

\begin{thm}[RANDOM soundness w.r.t. outlier treatment] \label{theorem: random}For any $R$, let $\llbracket \text{RANDOM} \; A\; l\; u \rrbracket (R, \rho, V) = (R',\rho')$, where $R' = (R-V) \cup V_r$, $V$ is the affected bag for a certain assertion, whereas $V_r$ is the corresponding replacing bag for the action RANDOM. Then, $\forall r\in V$,

$$\llbracket A \rrbracket _{\eta_{R',r'}} \neq \llbracket A \rrbracket _{\eta_{R,r}}$$

unless $\llbracket A \rrbracket _{\eta_{R,r}} \in [l,u]$ and $\text{random}(l, u, \rho) = \llbracket A \rrbracket _{\eta_{R,r}}$.
    
\end{thm}

\begin{proof} Let us choose an arbitrary record from the affected bag, $r \in V$. From the semantics of the action RANDOM $A \; l\; u$, $\llbracket A \rrbracket _{\eta_{R',r'}} = \text{random}(l, u, \rho)$ for some $\rho \in [0,1).$ We distinguish the following cases:

\begin{itemize}
    \item $\llbracket A \rrbracket _{\eta_{R,r}} \notin [l,u]$, then by the Definition~\ref{definition: random trace} of random trace, $\text{random}(l, u, \rho)=l+\lfloor \rho \cdot (u-l+1)\rfloor$, where $\rho \in [0,1)$, $\text{random}(l, u, \rho)\in [l,u]$ for all $\rho \in [0,1)$, hence $\llbracket A \rrbracket _{\eta_{R',r'}} \neq \llbracket A \rrbracket _{\eta_{R,r}}$;
    \item $\llbracket A \rrbracket _{\eta_{R,r}} \in [l,u]$, then $\llbracket A \rrbracket _{\eta_{R',r'}} = \llbracket A \rrbracket _{\eta_{R,r}}$ only if $\text{random}(l, u, \rho) = \llbracket A \rrbracket _{\eta_{R,r}}$ for some $\rho \in [0,1)$.
\end{itemize}
\end{proof}

\textbf{Interpretation of the Theorem~\ref{theorem: random}.} The point of the imputation method is to replace the outlier value with another value that reveals less information. Theorem~\ref{theorem: random} says that the replacing value $\llbracket A \rrbracket _{\eta_{R',r'}}$  will always be different from the original value of the outlier $\llbracket A \rrbracket _{\eta_{R,r}}$, unless the original value is already in the interval $[l,u]$ and $random(l,u,\rho)$ is equal to the original value. However, for the proper treatment of outliers, the data curator will never choose the lower and the upper bound for the action RANDOM so that the original value belongs to $[l, u]$, like in Example~\ref{example: outlier}. As a result, RANDOM is sound with respect to outlier treatment. 

\begin{thm}[k-anonymity soundness and completeness] \label{theorem: k-anonymity} For any $R$, let $\llbracket \text{REJECT} \rrbracket (R, \rho, V) = (R',\rho)$, where $R' = R-V $ and $V$ is the affected bag for the assertion representing the k-anonymity condition. Then:

\begin{itemize}
    \item[i)] $R'$ satisfies the k-anonymity principle.
    \item[ii)] For every $S\subseteq R$, if $S$ satisfies the k-anonymity principle then $S\subseteq R'$.
\end{itemize}
    
\end{thm}

\begin{proof} Let $QI=(Q_1, Q_2, \cdots, Q_m) \subset S$ be a tuple of quasi-identifiers for a relation $R$ and $R'$. By the Definition~\ref{definition: aff bag for EACH RESULT} of an affected bag and by the Definition~\ref{definition: k-anonymity semantical} of k-anonymity, we have

$$V=\left\{
\underbrace{r, \ldots, r }_{k \text{ times}}
\;\middle|\;
r\in_k R,\; |[r]_{QI,R}| < k
\right\}.$$

\begin{itemize}
    \item[i)] Let us choose an arbitrary $r \in R'$ and show that $|[r]_{QI,R'}| \geq k$. Since $r\in R'$ and $R'=R-V$, it holds $r\notin V$ and $|[r]_{QI,R}| \geq k$ (1). Also, $[r]_{QI,R}=[r]_{QI,R'}$ (2). From (1) and (2), we have $|[r]_{QI,R'}| \geq k$, which means that $R'$ satisfies the k-anonymity principle.
    \item[ii)] The proof is by contradiction. Let us choose $S$, an arbitrary sub-bag of $R$, which satisfies the k-anonymity principle, and $S \nsubseteq R'$. Then, there exists a record $r$ such that $r\in S$ and $r \notin R'$, meaning that $r\in 
    V$ and $|[r]_{QI,R}|<k$. Since $S\subseteq R$, $|[r]_{QI,S}| \leq |[r]_{QI,R}|<k$ which is the contradiction with the assumption that $S$ satisfies the k-anonymity principle. Therefore, $S\subseteq R'$ and $|S|\leq |R'|.$

\end{itemize}
    
\end{proof}

\textbf{Interpretation of the Theorem~\ref{theorem: k-anonymity}.} The theorem says that the semantics of REJECT action is: i) \textbf{sound} with respect to k-anonymity, i.e., all remaining records satisfy the k-anonymity principle, and ii) \textbf{complete} with respect to k-anonymity, i.e., $R'$ is the largest obtainable sub-bag of $R$ produced by REJECT.

\begin{thm}[l-diversity soundness and completeness] \label{theorem: l-diversity} For any $R$, let $\llbracket \text{REJECT} \rrbracket (R, \rho, V) = (R',\rho)$ where $R' = R-V$ and $V$ is the affected bag for the assertion representing the l-diversity condition. Then:

\begin{itemize}
    \item[i)] $R'$ satisfies the l-diversity principle.
    \item[ii)] For every $S\subseteq R$, if $S$ satisfies the l-diversity principle then $S\subseteq R'$.
\end{itemize}
\end{thm}

\begin{proof} Let $QI=(Q_1, Q_2, \cdots, Q_m) \subset S$ be a tuple of quasi-identifiers for a relation $R$ (and $R'$). Let $A$ be the sensitive attribute for the same relations. By the Definition~\ref{definition: aff bag for EACH RESULT} of an affected bag and by the Definition~\ref{definition: l-diversity semantical} of l-diversity, we have

$$V=\left\{
\underbrace{r, \ldots, r }_{k \text{ times}}
\;\middle|\;
r\in_k R,\; |\upvarepsilon (\{\llbracket A \rrbracket_{\eta_{S,r}}, r \in [r]_{QI,R}\})| < l
\right\}.$$

\begin{itemize}
    \item[i)] Let us choose an arbitrary $r \in R'$ and show that $|\upvarepsilon (\{\llbracket A \rrbracket_{\eta_{S,r}},r \in [r]_{QI,R'}\})| \geq l$. Since $r\in R'$ and $R'=R-V$, it holds $r\notin V$ and $|\upvarepsilon (\{\llbracket A \rrbracket_{\eta_{S,r}}, r \in [r]_{QI,R}\})| \geq l$ (1). Also, $[r]_{QI,R}=[r]_{QI,R'}$ (2). From (1) and (2), we have $ |\upvarepsilon (\{\llbracket A \rrbracket_{\eta_{S,r}},r \in [r]_{QI,R'}\})| \geq l$, which means that $R'$ satisfies the l-diversity principle.
    \item[ii)] The proof is by contradiction. Let us choose $S$, an arbitrary sub-bag of $R$, which satisfies the l-diversity principle, and $S \nsubseteq R'$. Then, there exists a record $r$ such that $r\in S$ and $r \notin R'$, meaning that $r\in V$ and $|\upvarepsilon (\{\llbracket A \rrbracket_{\eta_{S,r}},r \in [r]_{QI,R}\})| <l$. Since $S\subseteq R$, $|\upvarepsilon (\{\llbracket A \rrbracket_{\eta_{S,r}},r \in [r]_{QI,S}\})|  \leq |\upvarepsilon (\{\llbracket A \rrbracket_{\eta_{S,r}},r \in [r]_{QI,R}\})| <l$, which is the contradiction with the assumption that $S$ satisfies the l-diversity principle. Therefore, $S\subseteq R'$ and $|S|\leq |R'|.$

\end{itemize}
    
\end{proof}

\textbf{Interpretation of the Theorem~\ref{theorem: l-diversity}.} The theorem says that the semantics of REJECT action is: i) \textbf{sound} with respect to l-diversity, i.e., all remaining records satisfy the l-diversity principle), and ii) \textbf{complete} with respect to l-diversity, i.e., $R'$ is the largest obtainable sub-bag of $R$ produced by REJECT.

\section{Discussion and Conclusion}\label{conclusion}
We have extended the COMPASS language of \cite{GOKI2023}, which is used for anonymity analysis in databases. The new language is called A-COMPASS. Beside the one record - one group setting in which COMPASS works, A-COMPASS can also work in one record - one person setting, which is more suitable for microdata analysis. Further, A-COMPASS enables users to perform anonymization methods, such as suppression and generalization, not only to check whether a certain anonymity conditions is applied. A-COMPASS extends the COMPASS syntax with an action REPLACE A with a, and a new aggregation operator COUNT DISTINCT, in order to allow reasoning about k-anonymity and l-diversity. We have exemplified A-COMPASS features on a microdata set related to annual electricity consumption. We have also introduced a semantics for A-COMPASS, inspired by SQL semantics of \cite{GuLi2017}. Finally, we have proven the determinism and compositionality of the introduced semantics, as well as soundness and completeness of the introduced actions with respect to anonymity analysis.

Guagliardo and Libkin in \cite{GuLi2017} show the equivalence of their denotational semantics for SQL with relational algebra. Moreover, Palamidessi and Stronati in \cite{CatMar2012} investigate sensitivity bounds for differential privacy in relational algebra. Together, these two lines of research form a direction for future work. First, future work may investigate whether A-COMPASS requirements correspond to expressions in relational algebra. Second, in order to additionally allow privacy analysis, the extension of A-COMPASS may be considered. This approach also includes probabilistic RANDOM modeling so that the language can support working with concrete distributions and reasoning about differential privacy. Finally, the implementation of the language will be considered.

\bibliographystyle{plain}
\bibliography{bibliography}

\end{document}